\documentclass [12pt]{article} 
\usepackage{amssymb,amsmath,amsthm}
\usepackage{graphicx}
\usepackage[utf8]{inputenc}
\usepackage[T1]{fontenc}

\newtheorem{thm}{Theorem}[section]
\newtheorem{lemma}[thm]{Lemma}
\newtheorem{prop}[thm]{Proposition}
\newtheorem{cor}[thm]{Corollary}

\theoremstyle{remark}

\theoremstyle{remark}
\newtheorem*{note}{Remark}

\theoremstyle{definition}
\newtheorem{df}[thm]{Definition}
\newtheorem{ex}[thm]{Example}

\DeclareMathOperator{\tr}{Tr}

\newcommand{\N}{\mathbb{N}}

\newcommand{\R}{\mathbb{R}}

\newcommand{\il}{\langle}
\newcommand{\ir}{\rangle}

\begin{document}

\title{{\bf On orthogonal projections on the space of consistent pairwise comparisons matrices}} 

\author{Waldemar W. Koczkodaj
\thanks{Computer Science, Laurentian University, Sudbury, Ontario P3E 2C6, Canada, wkoczkodaj@cs.laurentian.ca}
\and
Ryszard Smarzewski \thanks{Institute of Mathematics and Cryptology
Cybernetics Faculty, Military University of Technology
Kaliskiego 2, 00-908 Warsaw, Poland, ryszard.smazewski@gmail.com}
\and
Jacek Szybowski \thanks{AGH University of Science and Technology, Faculty of Applied Mathematics, al. Mickiewicza 30, 30-059 Kraków, Poland, szybowsk@agh.edu.pl}
}

\maketitle

\begin{abstract}
In this study, the orthogonalization process for different inner products is applied to pairwise comparisons.
Properties of consistent approximations of a given inconsistent pairwise comparisons matrix are examined. A method of a derivation of a priority vector induced by a pairwise comparison matrix for a given inner product has been introduced.

The mathematical elegance of orthogonalization and its universal use in most applied sciences has been the motivating factor for this study. However, the finding of this study that approximations depend on the inner product assumed, is of considerable importance. 
\end{abstract}

\maketitle



\renewcommand{\thefootnote}{\arabic{footnote}} \setcounter{footnote}{0}

\section{Introduction}


The growing number of various orthogonalization approaches in \cite{FI2, FI4, FI1, FI3} supports the importance of orthogonalization in various computer science applications. Pairwise comparisons allow us to express assessments of many entities (especially, of the subjective nature) into one value for the use in the decision making process.  
Pairwise comparisons have been used since the late years in the 13th century by Llull for conducting the better election process (as stipulated in \cite{Llull_election}). 
However, the ineffability of pairwise comparisons comes from decision making which must have been made by our ancestors during the Stone Age.
Two stones must have been compared to decide which of them fit for the purpose. It could be for a hatchet, a gift, or a decoration. 


Pairwise comparisons matrices can be transformed by a logarithmic mapping into a linear space and the set of consistent matrices into its subspace. The structure of a Hilbert space is obtained by using an inner product.
Such a space is complete with respect to the norm corresponding to the inner product.
In such a space, we may use orthogonal projections as a tool to produce a consistent approximation of a given pairwise comparison matrix.
%
\subsection*{Structure of the paper}
A gentle introduction to pairwise comparisons is provided in Section~\ref{PCs}. Section~\ref{approx} discusses the problem of approximation of an inconsistent PC matrix by a consistent PC matrix
using Frobenius inner product on the space of matrices. Other inner products are discussed in Section~\ref{products}. In Section~\ref{crit} the dependence of an optimal priority vector on the choice of an inner product on the space of pairwise comparison matrices has been proved. The Conclusions are self explanatory.

\section{Pairwise comparisons matrices}
\label{PCs}
In this subsection, we define a pairwise comparisons matrix (for short, PC matrix) and introduce some related notions.
Pairwise comparisons are traditionally stored in a PC matrix. It is a square $n\times n$ matrix $M=[m_{ij}]$ with real positive elements $m_{ij}>0$ for every $i,j=1,\ldots,n,$ where $m_{ij}$ represents a relative preference of an entity $E_i$ over $E_j$ as a ratio. The entity could be an object, attribute of it, abstract concept, or a stimulus. For most abstract entities, we do not have a well established measure such as a meter or kilogram. ``Software safety'' or ``environmental friendliness'' are examples of such entities or attributes used in pairwise comparisons.
\par
When we use a linguistic expression containing "how many times", we process ratios.
The linguistic expression "by how much", "by how much percent" (or similar) gives us a relative difference.
Ratios often express subjective preferences of two entities, however, it does not imply that they can be obtained only by division.
In fact, equalizing the ratios with the division (e.g., $E_i/E_j$), for pairwise comparisons, is in general unacceptable.
It is only acceptable when applied to entities with the existing units of measure (e.g., distance).
However, when entities are subjective (e.g., reliability and robustness commonly used in a software development process
as product attributes), the division operation has no mathematical meaning although we can still consider which of them
is more (or less) important than the other for a given project.
The use of the symbol "/" is in the context of "related to" (not the division of two numbers).
Problems with some popular customization of PCs have been addressed in \cite{collab2}.
We decided not to address them here.   
\par
A PC matrix $M$ is called {\em reciprocal} if $m_{ij} = \frac{1}{m_{ji}}$ for every $i,j=1, \ldots ,n$.
In such case, $m_{ii}=1$ for every $i=1, \ldots ,n $.
\par
We can assume that the PC matrix has positive real entries and is reciprocal without the loss of generality
since a non reciprocal PC matrix can be made reciprocal by the theory presented in \cite{KO1999}.
The conversion is done by replacing $a_{ij}$ and $a_{ji}$ with geometric means of $a_{ij}$ and $a_{ji}$ ($\sqrt{a_{ij}a_{ji}}$).
The reciprocal value is $\frac{1}{\sqrt{a_{ij}a_{ji}}}$.
\par
Thus a PC matrix $M$ is the $n\times n$-matrix of the form:
 \begin{displaymath}
M = \begin{bmatrix}
1 & m_{12} & \cdots & m_{1n} \\
\frac{1}{m_{12}} & 1 & \cdots & m_{2n} \\
\vdots & \vdots & \vdots & \vdots \\
\frac{1}{m_{1n}} & \frac{1}{m_{2n}} & \cdots & 1
\end{bmatrix}.
\end{displaymath}
Sometimes, we write that $M\in\mathrm{PC}_n$ in order to indicate the size of a given PC matrix.

\subsection{The Geometric Means Method}
The main goal to use a pairwise comparison matrix is to obtain the so called priority vector. The coordinates of this vector correspond to the weights of alternatives. If we know the priority vector, we can set alternatives in order from the best to the worst one.

In the Geometric Means Method (GMM) introduced in \cite{CW1985} the coordinates of the vector are calculated as the geometric means of the elements in rows of the matrix:
\begin{equation}\label{GMM}
v_i=\sqrt[n]{\prod_{j=1}^n a_{ij}}.
\end{equation}
The above vector is the solution of the Logarithmic Least Square Method.

\subsection{Triads, transitivity, and submatrices of a PC matrix}
One of the fundamental problems in pairwise comparisons is the inconsistency.
It takes place when we provide, for any reason, all (hence supernumerary) comparisons of
$n$ entities which is $n^2$ or $\frac{n(n-1)}{2}$ if the reciprocity is assumed and used to reduce the number of entered comparisons.
The sufficient number of comparisons is $n-1$, as stipulated in \cite{KS2015}, but this number is based on some
arbitrary selection criteria of the minimal set of entities to compare. In practice, we have a tendency to make all
$n(n-1)/2$ comparisons (when reciprocity is assumed which is expressed by $m_{ij}=\frac{1}{m_{ji}}$ property also not
always without its problem). Surprisingly, the equality $x/y = \frac{1}{y/x}$ does not take place even if both
$x \neq 0$ and $y \neq 0$. For example, the blind wine testing may result in claiming that $x$ is better than $y$ and
$y$ is better than $x$ or even that $x$ is better than $x$ which is placed on the main diagonal in a PC matrix $M$,
expressing all pairwise comparisons in a form of a matrix.



The basic concept of inconsistency may be illustrated as follows. If an alternative $A$ is three times better than $B$, and $B$ is twice better than $C$, than $A$ should not be evaluated as five times better than C. Unfortunately, it does not imply that $A$ to $C$ should be $3\cdot 2$ hence 6, as the common sense may dictate, since all three assessments (3, 5, and 2) may be inaccurate and we do not know which one of them is or not incorrect. Inconsistency is sometimes mistakenly taken
for the approximation error but it is incorrect.
For example, triad $T=(3,5,2)$ can be approximated by $T_{\mathrm{approx}}(1,1,1)$ with 0 inconsistency but we can see that such approximation is far from optimal by any standard.
So, the inconsistency can be 0 yet the approximation error can be different than 0 and of arbitrarily large value.

\subsection{Multiplicative variant of pairwise comparisons}
\begin{df}
Given $n\in\N$, we define $$\mathcal{T}(n)=\{(i,j,k)\in\{1,\ldots,n\}:\, i<j<k \}$$
as the set of all PC matrix indexes of all permissible triads in the upper triangle.
\end{df}
\begin{df}\label{dftransitive}
A PC matrix $M=[m_{ij}]$ is called {\em consistent} (or {\em transitive}) if, for every $(i,j,k) \in \mathcal{T}(n):$
\begin{equation}\label{cc}
m_{ik}m_{kj}=m_{ij}.
\end{equation}
\end{df}
Equation~(\ref{cc}) was proposed a long time ago (in 1930s) and it is known as a "consistency condition".
Every consistent PC matrix is reciprocal, however, the converse is false in general.
If the consistency condition does not hold, the PC matrix is inconsistent (or intransitive).
In several studies, conducted between 1940 and 1961 (\cite{GS1958,H1953,KB1939,S1961})
the inconsistency in pairwise comparisons was defined and examined.
\par
Inconsistency in pairwise comparisons occurs due to superfluous input data. As demonstrated in \cite{KS2015}, only $n-1$ pairwise comparisons are really needed to create the entire PC matrix for $n$ entities, while the upper triangle has $n(n-1)/2$ comparisons. Inconsistencies are not necessarily "wrong" as they can be used to improve the data acquisition. However, there is a real necessity to have a "measure" for it.


\noindent 

\begin{lemma}\label{l23} If a $PC$ matrix $M=\ {\left[m_{ij}\right]}^n_{i,j=1}$ is consistent, then 

                                $$m_{ij}=\frac{{\omega }_i}{{\omega }_j}  \mbox{ for all } i,j=1,2,\dots ,n$$

\noindent where ${\omega }_1>0\ \ \mathrm{is\ arbitrary\ and}\ \ {\omega }_j=\frac{{\omega }_1}{m_{1j}}$ for every $j=2,3,\dots ,n$.

\end{lemma}

\begin{proof} By the definition of ${\omega }_j$ and consistency of $M$ one gets

                    $$m_{1j}=\frac{{\omega }_1}{{{\omega }_1}/{m_{1j}}}=\frac{{\omega }_1}{{\omega }_j}$$ and  $$m_{ij}=\frac{m_{i-1,j}}{m_{i-1,i}}=\frac{{{\omega }_{i-1}}/{{\omega }_j}}{{{\omega }_{i-1}}/{{\omega }_i}}=\frac{{\omega }_i}{{\omega }_j}.$$

\noindent whenever $1<i\le n.$   
\end{proof} 

\noindent 

It is easy to observe that the set ${\mathcal{M}}_n=\left({\mathcal{M}}_n, \cdot \right)$ of all consistent $PC$ 
matrices $M$ is a multiplicative subgroup of the group of all $PC$ $n\times n$-matrices endowed with the coordinate-wise multiplication $A \cdot B=\left[a_{ij}b_{ij}\right],$ where $A=\left[a_{ij}\right]$ and  $B=\left[b_{ij}\right]$. Its representation in ${\mathbb{R}}^n$ consists of all priority vectors $\upsilon \left(M\right)=$\linebreak $\left({\omega }_1,{\omega }_2,\dots ,{\omega }_n\right),$ defined uniquely as in Lemma \ref{l23}, up to a multiplicative constant ${\omega }_1>0.$ In the following we use priority vectors normalized by the condition  ${\omega }_1=1,$ unless otherwise stated.  

\noindent 

\noindent \textbf{}

\subsection{Additive variant of pairwise comparisons}
Instead of a PC matrix $M=[m_{ij}]$ with $m_{ij}\in\R_+^*,$ the set of positive real numbers considered with multiplication, we can transform entries of $M$ by a logarithmic function and get a matrix $A=\left[a_{ij}\right]=\left[{\mathrm{log}\ m_{ij}}\right].$ Since a $PC$  matrix $M$ is reciprocal, it follows that it is anti-symmetric, i.e.

$$a_{ij}=-a_{ji}\mbox{ for every }i,j=1,2,\dots ,n.$$          

\noindent Moreover, if $M$ is consistent then $A={\mathrm{log}\ M}$ satisfies  the condition of additive  consistency:

       $$a_{ik}+a_{kj}=a_{ij}\mbox{ for every }\left(i,j,k\right)\in \mathcal{T}\left(n\right),$$   

\noindent which yields the following well-known representation.    

\noindent                                                      

\begin{lemma}\label{l24} If an  anti-symmetric matrix $A=\ {\left[a_{ij}\right]}^n_{i,j=1}$ is additively consistent, then 

                                $$a_{ij}={\sigma }_i-{\sigma }_j\mbox{ for all }i,j=1,2,\dots ,n,$$

\noindent where  ${\sigma }_1$ is arbitrary and ${\sigma }_j={\sigma }_1-$ $a_{1j}$ for every $j=2,3,\dots ,n$. 

\end{lemma}

 In view of this representation, the set ${\mathcal{A}}_n=\left({\mathcal{A}}_n,+\right)$ of all additively consistent matrices  is an additive subgroup of all $n\times n$- matrices, whenever it is endowed with the coordinatewise matrix addition  $A+B=\left[a_{ij}{+b}_{ij}\right]\ $ of $A=\left[a_{ij}\right]$ and  $B=\left[b_{ij}\right]$. It is a one-to-one image of the  multiplicative group ${\mathcal{M}}_n=\left({\mathcal{M}}_n,\cdot \right)$  by the group isomorphism $A={\mathrm{log}\ M}=[{\mathrm{log}\ m_{ij}]\ }$. The inverse group isomorphism is clearly given by the formula $M={\mathrm{exp}\ A}=\left[{\mathrm{exp}\ a_{ij}}\right].$ Moreover, the additive priority vector $\upsilon \left(A\right)=\left({\sigma }_1,{\sigma }_2,\dots ,{\sigma }_n\right)$ of $A$ satisfies $\upsilon \left(A\right)={\mathrm{log} \ \upsilon \left(M\right),\ }$ where ${\sigma }_1={\mathrm{log} {\omega }_1\ }$ is supposed to be  arbitrary  additive constant. In particular, it is said to be normalized if ${\sigma }_1=0.$ Here and in the following matrix functions ${\mathrm{log}\ M=\left[{\mathrm{log}\ m_{ij}}\right]\ }$ and ${\mathrm{exp}\ A\ }=[{\mathrm{exp}\ a_{ij}]\ }$ are always understood in the coordinate- wise sense.


\section{Approximation by projections}
\label{approx}

Numerous heuristics have been proposed for approximations of inconsistent pairwise comparisons matrices by consistent pairwise comparisons matrices. Geometric means (GM) of rows is regarded as dominant. Some mathematical evidence, to support GM as the method of choice, was also provided in \cite{CW2D016}.
\cite{KS2016} shows that orthogonal projections have a limit which is GM (to a constant). \cite{KKSX} demonstrates that the inconsistency reduction algorithm based on the orthogonal projections converges very quickly for practical applications. The proof of inconsistency convergence was outlined in \cite{HK1996} and finalized in \cite{KS2010}. Axiomatization of inconsistency still remains elusive. Its recent mutation in \cite{KU17} has a deficiency (the  monotonicity axiom incorrectly defined).
 
\subsection{Space of consistent matrices}
\noindent Let $\mathbb{K}$= $\mathbb{R}\ \mathrm{or}\mathrm{\ }\mathbb{C}\mathrm{.}$ Let$\ M\left(n,\mathbb{K}\right)$ be the set of all$\ n \times n$-matrices with entries    from the field$\ \mathbb{K},$ and let  ${\mathcal{C}=\mathcal{M}}_n\subset M\left(n,\mathbb{K}\right)$ be the set of all consistent $n\times n$-matrices with entries from the field$\ \mathbb{K}$. We consider $M\left(n,\mathbb{K}\right)$ as a $\mathbb{K}-$linear space with addition of matrices and multiplication by numbers from the field $\mathbb{K},$  clearly  ${dim}_{\mathbb{K}}\ M\left(n,\mathbb{K}\right)$=\textbf{ }$n^2$\textbf{ }and\textbf{ }the unit\textbf{ }matrices 
\[E_{ij}={\left[e^{i,j}_{rs}\right]}^n_{r,s=1},\ \ i,j=1,2,\dots ,n,\]

\noindent form a basis in $M\left(n,\mathbb{K}\right),$ where  $e^{i,j}_{rs}$ is equal to 1, if $r=i$ and $s=j,$ and   otherwise 0.  

In the linear space $M\left(n,\mathbb{K}\right)\ $ one can define the Frobenius inner product as follows. For all $A=\left[a_{ij}\right],B=\left[b_{ij}\right]\in M\left(n,\mathbb{K}\right),$ \textbf{}

\begin{equation*}
\langle A,B\rangle_{\mathrm{F}}=\sum_{i=1}^n\sum_{j=1}^n a_{ij}\bar{b}_{ij}.
\end{equation*}
In this Section we recall results from \cite{KO1997}.
\begin{thm}
The set $\mathcal C$ is a linear subspace of $\mathrm{M}(n,\mathbb K).$
\end{thm}
\begin{proof}
Let $A=[a_{ij}],B=[b_{ij}]\in\mathcal C,$ that is
$$a_{ik}+a_{kj}=a_{ij}\text{ and }b_{ik}+b_{kj}=b_{ij}.$$
Let $C=[c_{ij}]=A+B,$ then
$$c_{ik}+c_{kj}=(a_{ik}+b_{ik})+(a_{kj}+b_{kj})=(a_{ik}+a_{kj})+(b_{ik}+b_{kj})=a_{ij}+b_{ij}=c_{ij}.$$
Hence, $C\in\mathcal C.$
\par
Let $\alpha\in\mathbb K$ and $A\in\mathcal C.$ It is clear that $\alpha A\in\mathcal C.$
\end{proof}
\begin{thm}\label{dimension}
The subspace $\mathcal C\subset\mathrm{M}(n,\mathbb K)$ has dimension $n-1$ over $\mathbb K.$
\end{thm}
\begin{proof}
By applying the consistency condition, all elements of the matrix $A=[a_{ij}]$ can be generated by $n-1$ elements $a_{k,k+1}$ for $k=1,\ldots,n-1,$ i.e. by the second diagonal, that is diagonal directly above the main diagonal (see \cite{KS2015}).
\end{proof}
\begin{thm}[{\cite[Proposition 1]{KO1997}}]\label{basis}
The following set of $n-1$ matrices constitutes a basis of $\mathcal C:$
\begin{equation*}
B_k=[b_{ij}^k],\text{ where }b_{ij}^k=
\begin{cases}
\phantom{-}1,&\text{for }1\leq i\leq k<j\leq n,\\
-1,&\text{for }1\leq j\leq k<i\leq n,\\
\phantom{-}0,&\text{otherwise,}
\end{cases}
\end{equation*}
where $k=1,\ldots n-1.$
\end{thm}
\begin{note}
For the standard inner product (i.e. Frobenius), an example of approximation of a $4\times 4$ inconsistent matrix as a projection onto $\mathcal C$ is given in \cite
{KO1997}.
\end{note}

\subsection{Approximation by a consistent matrix}

\noindent Suppose that we have a $PC$ matrix $A\in M\left(n,\mathbb{K}\right)\backslash \mathcal{C},$ i.e. $A$ is  inconsistent. Our aim is to find a consistent metric projection $A_{\mathcal{C}}$ of $A\ $ onto the set $\ \mathcal{C}={{\mathcal{A}}_n\ \mathrm{or\ }\mathcal{M}}_n$ with respect to norm $\left\|\cdot \right\|$ induced by an inner product $\left\langle \cdot ,\cdot \right\rangle $, i.e. a nonlinear mapping  $A_{\mathcal{C}}\ :M\left(n,\mathbb{K}\right) \ni A\mapsto A_{approx} \in {\mathcal{C}}$ such that the distance of $A$ to  $\mathcal{C}$ 

\noindent 
\[\ dist\left(A,\ \mathcal{C}\right)={\inf}_{B\in \mathcal{C}}\left\|A-B\right\|=\left\|A-A_{approx}\right\|.  \]

\noindent is attained by the matrix $B=A_{approx}.$

\noindent 

In the additive case $\mathcal{C}={\mathcal{A}}_n$ metric projection $A_{{\mathcal{A}}_n}$ coincides with the  orthogonal projection $A_{proj}:A\mapsto A_{approx}$ of $M\left(n,\mathbb{K}\right)$ onto the $\left(n-1\right)$-linear subspace ${\mathcal{A}}_n,$  which is characterized by the well-known orthogonality condition

\noindent 
\[A-A_{approx}\bot \ {\mathcal{A}}_n.                     \]

\noindent This condition enables to compute the orthogonal projection $A_{proj}$ much more effectively than its nonlinear multiplicative counterpart $M_{{\mathcal{M}}_n}:M\mapsto M_{approx}$. Therefore, it was proposed \cite{CW1985,KS2016} to linearize the process of determining metric projections for practical applications. It was achieved by introducing a new concept of linearized consistent approximations to estimate nonlinear metric projections. For the simplicity, in the following  the symbol $M_{approx}$ will be also used to denote these  linearized consistent approximations. It would not lead to misunderstanding, since we shall always restrict our attention to the linearized case, unless otherwise stated. 

\noindent 

\begin{df}\label{aprox} 
Let $M\in M\left(n,\mathbb{K}\right)\backslash {\mathcal{M}}_n$ be a $PC$ inconsistent  matrix. 

\noindent A consistent approximation $M_{proj}:M\mapsto M_{approx}$ of $M$ onto ${\mathcal{M}}_n$ is defined in the following way:

\begin{enumerate}
\item  we construct the matrix $A={\mathrm{log}\ M,}$

\item  we find the orthogonal projection $A_{approx}$ of $A$ onto the $\left(n-1\right)$-dimensional subspace $\ \mathcal{C}={\mathrm{log}\ {\mathcal{M}}_n\ }.$ 

\item  we set   $M_{approx}={\mathrm{exp} \left(A_{approx}\right)\ }.$
\end{enumerate}

\noindent In short, we define $M_{approx}={\mathrm{exp} \left[{({\mathrm{log}\ M)}}_{approx}\right].}$ 
\end{df}

\subsection{Orthogonalization}
In order to simplify calculation in the examples below, we would like to have orthogonal basis for $\mathcal{C}.$ We produce such a basis by the Gram-Schmidt process. Namely, let $V$ be an $n$-dimensional vector space over $\mathbb K$ with an inner product $\il\cdot,\cdot\ir$ and $B_1,\ldots, B_n$ be its basis. We construct an orthogonal basis $E_1,\ldots,E_n$ as follows:
\begin{equation}\label{GS}
\begin{split}
E_1=&B_1,\\
E_2=&B_2-\frac{\il E_1,B_2\ir}{\il E_1,E_1\ir}E_1,\\
E_3=&B_3-\frac{\il E_1,B_3\ir}{\il E_1,E_1\ir}E_1-\frac{\il E_2,B_3\ir}{\il E_2,E_2\ir}E_2,\\
\ldots=&\ldots\\
E_n=&B_n-\sum_{j=1}^{n-1}\frac{\il E_j,B_n\ir}{\il E_j,E_j\ir}E_j.
\end{split}
\end{equation}

\begin{ex}\label{exF}
Consider an inconsistent PC matrix $M$ in the multiplicative variant:
\begin{equation}\label{M}
M=
\begin{bmatrix}
1&e^{2}&e^{7}\\
e^{-2}&1&e^{3}\\
e^{-7}&e^{-3}&1
\end{bmatrix}.
\end{equation}
Its priority vector $v(M)$ obtained by (\ref{GMM}) is
\begin  {equation}\label{v1}
v(M)=
\begin{bmatrix}
e^3\\
e^{\frac{1}{3}}\\
e^{-\frac{10}{3}}
\end{bmatrix}.
\end{equation}
Taking natural logarithms, we switch to the additive PC matrix variant and get the following additive PC matrix:
\begin{equation*}
A=
\begin{bmatrix}
0&2&7\\
-2&0&3\\
-7&-3&0
\end{bmatrix}.
\end{equation*}
We need to find $A_{\textrm{proj}}$, the projection of $A$ onto $\mathcal C.$ By Theorem~\ref{dimension}, we have that
$\dim_\R\mathcal C=2.$ By Theorem~\ref{basis}, we get a basis of the linear space of consistent matrices $\mathcal C:$
\begin{equation*}
B_1=
\begin{bmatrix}
\phantom{-}0&1&1\\
-1&0&0\\
-1&0&0
\end{bmatrix}
\text{ and }
B_2=
\begin{bmatrix}
\phantom{-}0&\phantom{-}0&1\\
\phantom{-}0&\phantom{-}0&1\\
-1&-1&0
\end{bmatrix}.
\end{equation*}
Evidently, $\il B_1,B_2\ir_{\mathrm{F}}=2.$ Therefore, we have to apply Gram-Schmidt process of orthogonalization \eqref{GS}. If $E_1,E_2$ denotes an orthogonal basis of $\mathcal C,$ then
\begin{equation*}
E_1=
\begin{bmatrix}
\phantom{-}0&1&1\\
-1&0&0\\
-1&0&0
\end{bmatrix}
\text{ and }
E_2=\begin{bmatrix}
\phantom{-}0&-\frac{1}{2}&\frac{1}{2}\\
\phantom{-}\frac{1}{2}&\phantom{-}0&1\\
-\frac{1}{2}&-1&0
\end{bmatrix}.
\end{equation*}
Our goal is to find $A_{\mathrm{proj}}=\varepsilon_1E_1+\varepsilon_2E_2,$ that is to find coefficients $\varepsilon_1$ and $\varepsilon_2$ such that
for every $C\in\mathcal C,$  $\il A-A_{\mathrm{proj}},C\ir_{\mathrm{F}}=0$ which is equivalent to solving:
\begin{equation*}
\begin{split}
\il A-\varepsilon_1E_1-\varepsilon_2E_2,E_1\ir_{\mathrm{F}}=&0,\\
\il A-\varepsilon_1E_1-\varepsilon_2E_2,E_2\ir_{\mathrm{F}}=&0.
\end{split}
\end{equation*}
Since $E_1$ and $E_2$ are orthogonal, we get a system of linear equations:
\begin{equation*}
\begin{split}
\il A,E_1\ir_{\mathrm{F}}-\varepsilon_1\il E_1,E_1\ir_{\mathrm{F}}=&0,\\
\il A,E_2\ir_{\mathrm{F}}-\varepsilon_2\il E_2,E_2\ir_{\mathrm{F}}=&0.
\end{split}
\end{equation*}
By computing Frobenius inner products, we get the following equation:
\begin{equation*}
\begin{split}
18-4\varepsilon_1=&0,\\
11-3\varepsilon_2=&0.
\end{split}
\end{equation*}
By solving the above equations for $\varepsilon_1,\varepsilon_2,$ we get $\varepsilon_1=\frac{9}{2}$ and $\varepsilon_2=\frac{11}{3}.$ Thus,
\begin{equation*}
A_{\mathrm{proj}}=A_{\mathrm{approx,F}}=\frac{9}{2}E_1+\frac{11}{3}E_2=
\begin{bmatrix}
\phantom{-}0&\frac{8}{3}&\frac{19}{3}\\
-\frac{8}{3}&\phantom{-}0&\frac{11}{3}\\
-\frac{19}{3}&-\frac{11}{3}&0
\end{bmatrix}.
\end{equation*}
Finally, we get a consistent approximation for $M,$
\begin{equation*}
M_{\mathrm{approx,F}}=
\begin{bmatrix}
1&e^{\frac{8}{3}}&e^{\frac{19}{3}}\\
e^{-\frac{8}{3}}&1&e^{\frac{11}{3}}\\
e^{-\frac{19}{3}}&e^{-\frac{11}{3}}&1
\end{bmatrix}\in\mathcal C.
\end{equation*}
Notice that the priority vector $v(M_{\mathrm{approx,F}})$ coincides with $v(M)$ given by (\ref{v1}).
\end{ex}

\section{Other inner products on $\mathrm{M}(n,\mathbb K)$}
\label{products}
The standard (Frobenius) inner product on the linear space $\mathrm{M}(n,\mathbb K)$ is defined by:
\begin{equation}\label{ilF}
\il A,B\ir_{\mathrm{F}}=\tr(B^*A).
\end{equation}
The above inner product is exactly the Frobenius inner product defined in previous section, and it defines the Frobenius norm in a usual way by:
$$\|A\|_{\mathrm{F}}^2=\il A,A\ir_{\mathrm{F}}=\sum_{i=1}^n\sum_{j=1}^n|a_{ij}|^2.$$
In \cite{scalar1} the following result is mentioned:
\begin{prop}\label{prop}
For every $m\in\N$ and positive semi-definite matrices $X_i,Y_i$, $i=1,\ldots,m,$ the following function:
\begin{equation}\label{general}
\il A,B\ir_*=\tr\left(\sum_{i=1}^mB^*X_iAY_i\right)
\end{equation}
defines an inner product in $\mathrm{M}(n,\mathbb K).$
\end{prop}
\begin{proof}
All properties of an inner product follow from the following equation:
\begin{multline*}
\il A,B\ir_*\\
=\tr\left(\sum_{i=1}^mB^*X_iAY_i\right)=\tr B^*\left(\sum_{i=1}^mX_iAY_i\right)=\left\il\sum_{i=1}^mX_iAY_i,B^*\right\ir_{\mathrm{F}}.
\end{multline*}
\end{proof}

\begin{ex}\label{dip}
Consider the following four matrices in the space $\mathrm{M}(3,\R):$
$$X_1=\begin{bmatrix}1&1&2\\1&2&3\\2&3&6\end{bmatrix},\;X_2=\begin{bmatrix}2&1&1\\1&2&1\\1&1&5\end{bmatrix},\;Y_1=\begin{bmatrix}2&3&2\\3&7&3\\2&3&5\end{bmatrix},
Y_2=\begin{bmatrix}5&2&1\\2&5&1\\1&1&1\end{bmatrix}.$$
By applying Sylvester's criterion in \cite{Lang}, it is easy to see that they are positive semi-definite. Evidently, they are symmetric hence Hermitian.
\par
Let
$$\mathcal A(A)=\mathcal A_{\{X_i,Y_i\mid i=1,2\}}(A)
=X_1AY_1+X_2AY_2.$$
Define
$$\il A,B\ir_{*}=\il\mathcal A(A),B\ir_{\mathrm{F}}.$$
By Proposition~\ref{prop}, $\il \cdot,\cdot\ir_{*}$ is an inner product in $\mathrm{M}(n,\R).$
\end{ex}
\begin{ex}\label{approx*}
Consider $3\times 3$ matrices $B_1$, $B_2$ with real entries computed by the formula in Theorem~\ref{basis} (see Example~\ref{exF} for details). Evidently, $\mathcal B=\{B_1,B_2\}$ is a basis for $\mathcal C\subset\mathrm{M}(3,\R).$ By applying Gram-Schmidt process \eqref{GS} with the inner product from Example~\ref{dip} to the basis $\mathcal B$, we get an orthogonal basis $\mathcal E=\{E_1,E_2\}$ for $\mathcal C$ in $\il\cdot,\cdot\ir_{*}.$
\par
\noindent The above transformations imply that $\il E_1,B_2\ir_*=\il\mathcal A(E_1),B_2\ir_{\mathrm{F}}$. 
\noindent Since
$$\mathcal A(E_1)=\begin{bmatrix}
-5&\phantom{-}9&\phantom{-}4\\
-17&-5&-3\\
-35&-13&-6\end{bmatrix},$$
\end{ex}
\noindent we have $$\il E_1,B_2\ir_*=49\text{ and }\il E_1,E_1\ir_*=65.$$
\noindent By equations \eqref{GS}, we get
$$E_1=
\begin{bmatrix}
\phantom{-}0&1&1\\
-1&0&0\\
-1&0&0
\end{bmatrix}
\text{ and }
E_2=
\begin{bmatrix}
\phantom{-}0&-\frac{49}{65}&\frac{16}{65}\\
\phantom{-}\frac{49}{65}&\phantom{-}0&1\\
-\frac{16}{65}&-1&\phantom{-}0
\end{bmatrix}=
\frac{1}{65}\begin{bmatrix}
\phantom{-}0&-49&16\\
\phantom{-}49&\phantom{-}0&65\\
-16&-65&\phantom{-}0
\end{bmatrix}.
$$
\begin{ex}\label{Aproj}
Take the following additive PC matrix:
\begin{equation*}
A=
\begin{bmatrix}
0&2&7\\
-2&0&3\\
-7&-3&0
\end{bmatrix}.
\end{equation*}
This is the PC matrix from Example~\ref{exF}. Next, we compute the orthogonal (with respect to the inner product from Example~\ref{dip}) projection onto the space $\mathcal C.$
For it, we need to solve a system of linear equations for $\varepsilon_1$ and $\varepsilon_2$:
\begin{equation}\label{one}
\begin{split}
\il A,E_1\ir_{*}-\varepsilon_1\il E_1,E_1\ir_{*}=&0,\\
\il A,E_2\ir_{*}-\varepsilon_2\il E_2,E_2\ir_{*}=&0.
\end{split}
\end{equation}
We get
$$\mathcal A(E_2)=-\frac{1}{65}
\begin{bmatrix}
570&1,710&\phantom{-}36\\
226&1,234&-650\\
1,930&5,050&\phantom{-}1,006
\end{bmatrix}.$$
We can also utilize some computation conducted in the previous example and by using the symmetry of the inner product $\il\cdot,\cdot\ir_*,$ the equation \eqref{one} becomes:
\begin{equation*}
\begin{split}
355-65\,\varepsilon_1=&0,\\
\frac{27,390}{65}-\left(\frac{1}{65}\right)^2 473,520\,\varepsilon_2=&0.
\end{split}
\end{equation*}
Consequently, $\varepsilon_1=\frac{355}{65}=\frac{71}{13},$ $\varepsilon_2=\frac{59,345}{15,784}$ therefore, we get:
$$A_{\mathrm{proj,1}}=A_{\mathrm{approx},*1}=\varepsilon_1E_1+\varepsilon_2E_2=\begin{bmatrix}
0&\varepsilon_1-\frac{49}{65}\varepsilon_2&\varepsilon_1+\frac{16}{65}\varepsilon_2\\
-\varepsilon_1+\frac{49}{65}\varepsilon_2&0&\varepsilon_2\\
-\varepsilon_1-\frac{16}{65}\varepsilon_2&-\varepsilon_2&0
\end{bmatrix}.
$$
Finally, we obtain the following multiplicative PC matrix:
$$M_{\mathrm{approx},*1}=\begin{bmatrix}
1&e^{\varepsilon_1-\frac{49}{65}\varepsilon_2}&e^{\varepsilon_1+\frac{16}{65}\varepsilon_2}\\
e^{-\varepsilon_1+\frac{49}{65}\varepsilon_2}&1&e^{\varepsilon_2}\\
e^{-\varepsilon_1-\frac{16}{65}\varepsilon_2}&e^{-\varepsilon_2}&1
\end{bmatrix}.
$$

\end{ex}

\begin{ex}\label{proj2}
Let us repeat the calculations made in Examples \ref{dip}, \ref{approx*} and \ref{Aproj} to provide a consistent approximation of the matrix $M$ set in (\ref{M}) by means of the inner product induced by matrices:
$$X_1=\begin{bmatrix}1&0&0\\0&2&0\\0&0&3\end{bmatrix},\;X_2=\begin{bmatrix}2&0&0\\0&3&0\\0&0&1\end{bmatrix},\;Y_1=\begin{bmatrix}3&0&0\\0&1&0\\0&0&2\end{bmatrix},
Y_2=\begin{bmatrix}1&0&0\\0&3&0\\0&0&2\end{bmatrix}.$$
We obtain 
$$\mathcal A(E_1)=\begin{bmatrix}
\phantom{-}0&7&6\\
-9&0&0\\
-10&0&0\end{bmatrix},$$
\noindent so $$\il E_1,B_2\ir_*=16\text{ and }\il E_1,E_1\ir_*=32.$$
\noindent By equations \eqref{GS}, we get
$$E_1=
\begin{bmatrix}
\phantom{-}0&1&1\\
-1&0&0\\
-1&0&0
\end{bmatrix}
\text{ and }
E_2=
\begin{bmatrix}
\phantom{-}0&-\frac{1}{2}&\frac{1}{2}\\
\phantom{-}\frac{1}{2}&\phantom{-}0&1\\
-\frac{1}{2}&-1&0
\end{bmatrix}=
\frac{1}{2}\begin{bmatrix}
\phantom{-}0&-1&1\\
\phantom{-}1&\phantom{-}0&2\\
-1&-2&0
\end{bmatrix}.
$$
Since
$$\mathcal A(E_2)=-\frac{1}{2}
\begin{bmatrix}
\phantom{-}0&-7&6\\
\phantom{-}9&\phantom{-}0&20\\
-10&-12&0
\end{bmatrix},$$
we calculate the inner products
$$\il A,E_1\ir_*=144,\ \il A,E_2\ir_*=88\text{ and }\il E_2,E_2\ir_*=24.$$
By solving the equations
\begin{equation*}
\begin{split}
144-32\,\varepsilon_1=&0,\\
88-24\varepsilon_2=&0.
\end{split}
\end{equation*}
we get $\varepsilon_1=\frac{9}{2},$ and $\varepsilon_2=\frac{11}{3}$ therefore, 
$$A_{\mathrm{proj,2}}=A_{\mathrm{approx},*2}=\varepsilon_1E_1+\varepsilon_2E_2=\begin{bmatrix}
\phantom{-}0&\phantom{-}\frac{8}{3}&\frac{19}{3}\\
-\frac{8}{3}&\phantom{-}0&\frac{11}{3}\\
-\frac{19}{3}&-\frac{11}{3}&0
\end{bmatrix}.
$$
Finally, 
$$M_{\mathrm{approx},*2}=\begin{bmatrix}
1&e^\frac{8}{3}&e^\frac{19}{3}\\
e^{-\frac{8}{3}}&1&e^\frac{11}{3}\\
e^{-\frac{19}{3}}&e^{-\frac{11}{3}}&1
\end{bmatrix},
$$
and its priority vector calculated with the use of GMM is equal to
\begin{equation*}
v(M_{\mathrm{approx},*2})=
\begin{bmatrix}
e^3\\
e^{\frac{1}{3}}\\
e^{-\frac{10}{3}}
\end{bmatrix}=v(M).
\end{equation*}

\end{ex}

\section{Approximation selection}\label{crit}
It is worthwhile to stress that in the previous examples we got three approximations of the same matrix $M$. An important dilemma has surfaced: how to compare different approximations of a given PC matrix obtained by the use of different inner products?
The answer to this question is: they are incomparable.

\subsection{Inconsistency} The first criterion that we took into consideration was to compare inconsistency indices of the exponential transformations of differences $A-A_{\mathrm{proj}}$. However, this attempt appeared to be incorrect. 

Let us consider the inconsistency index $Kii$ of a pairwise comparison matrix $M$ given by formula:
\begin{equation}\label{KiiG}
Kii(M)=\max_{i<j<k} \left(  1-\min\left\{\frac{m_{ik}}{m_{ij}m_{jk}},
\frac{m_{ij}m_{jk}}{m_{ik}}\right\}\right).
\end{equation}
This indicator satisfies all the desired axioms formulated in \cite{KU17}.

\begin{thm}
Let $A$ and $B$ be additive pairwise comparison matrices such that $B$ is additively consistent. Then
$$Kii\left(\mathrm{exp}(A-B)\right)=Kii\left(\mathrm{exp}(A)\right).$$
\end{thm}
\begin{proof}
Take any $(i,j,k) \in \mathcal{T}(n)$. Since $b_{ij}+b_{jk}=b_{ik}$, we get
\begin{eqnarray*}
1&-&\min\left\{\frac{e^{a_{ik}-b_{ik}}}{e^{a_{ij}-b_{ij}}e^{a_{jk}-b_{jk}}},
\frac{e^{a_{ij}-b_{ij}}e^{a_{jk}-b_{jk}}}{e^{a_{ik}-b_{ik}}}\right\}=\\
1&-&\min\left\{\frac{e^{a_{ik}}e^{b_{ij}+b_{jk}-b_{ik}}}{e^{a_{ij}}e^{a_{jk}}},
\frac{e^{a_{ij}}e^{a_{jk}}}{e^{a_{ik}}e^{b_{ij}+b_{jk}-b_{ik}}}\right\}=\\
1&-&\min\left\{\frac{e^{a_{ik}}}{e^{a_{ij}}e^{a_{jk}}},
\frac{e^{a_{ij}}e^{a_{jk}}}{e^{a_{ik}}}\right\},
\end{eqnarray*}
which completes the proof.
\end{proof}

From the above theorem it follows that if we take two different consistent approximations $B$ and $C$ of an additive matrix $A$ they satisfy $$Kii\left(\mathrm{exp}(A-B)\right)=Kii\left(\mathrm{exp}(A)\right)=Kii\left(\mathrm{exp}(A-C)\right).$$  

\subsection{Priority vectors for different inner products}
The second attempt to judge whether a consistent approximation $A_{approx}$ of a $PC$ matrix $A$ is acceptable could be to compare the priority vectors induced  by $A$ and $A_{approx}$ for any inner product. In \cite{CW1985} it has been proved that the elements of a projection matrix $A_{\mathrm{approx}}$ induced by a Frobenius product are given by the ratios $\frac{w_i}{w_j}$, where vector $w$ is obtained by GMM. As it has been shown in \cite{KS2016} the priority vectors induced  by $A$ and $A_{approx}$ in this case coincide:

\begin{thm}\label{th52}
Let $A$ be a PC matrix and $A_{\mathrm{approx}}=\left[\frac{w_i}{w_j}\right]$, where $w=GM(A)$, i.e.
$$w_k=\sqrt[n]{\prod_{j=1}^n a_{kj}}.$$
Then $GM(A)=GM(A_{\mathrm{approx}})$.
\end{thm}

As the following example shows the priority vectors of a matrix and its consistent approximation may differ if we use other inner products.

\begin{ex}
Consider an inconsistent additive PC matrix $A$ from Example~\ref{exF}:
\begin{equation*}
A=
\begin{bmatrix}
0&2&7\\
-2&0&3\\
-7&-3&0
\end{bmatrix}
\end{equation*}
and its corresponding multiplicative PC matrix $M=\mathrm{exp}(A)$.
Let us take three inner products: Frobenius product and the inner products $\il\cdot,\cdot\ir_{*1}$ and  $\il\cdot,\cdot\ir_{*2}$  from Examples~\ref{dip} and ~\ref{proj2}.
The approximations $A_{\mathrm{approx,F}}$, $A_{\mathrm{approx},*1}$ and $A_{\mathrm{approx},*2}$ are given in Examples ~\ref{exF}, ~\ref{approx*} and ~\ref{proj2}, respectively.

Notice that $$GM(\mathrm{exp}(A))=GM(\mathrm{exp}(A_{\mathrm{approx,F}}))=GM(\mathrm{exp}(A_{\mathrm{approx},*2})),$$
but $GM(\mathrm{exp}(A))$ and $GM(\mathrm{exp}(A_{\mathrm{approx},*1})),$ are linearly independent. This observation, however, is not surprising. The matrix $\mathrm{exp}(A_{\mathrm{approx},*1})$ minimizes the distance from $\mathrm{exp}(A)$ to the set of cosistent PC matrices according to the inner product $<\cdot,\cdot>_{*1}$, but not to the Frobenius inner product.
\end{ex}

\noindent In the following we show that as we change the inner product, we also have to change the formula for a priority vector. It is done by extending  Theorem 5.2 to weighted Frobenius inner products. For this purpose  we recall the most general standard definition of an inner product in $\mathbb{M}\left(n,\mathbb{K}\right):$ 

Let  $G_1,G_2,\dots ,G_N\ $be $N=n^2$ linearly  independent matrices in the space $\mathbb{M}\left(n,\mathbb{K}\right)$. Represent matrices $A,B\ \in \ \mathbb{M}\left(n,\mathbb{K}\right)$ in a unique manner as

$$A{=}\sum^{{N}}_{k=1}{{\alpha }_kG_k};\; {\alpha }_k  \in \mathbb{K},$$
and    
$$B{=}\sum^{{N}}_{{k}{=}{1}}{{\beta }_k}G_k;\;{\beta }_k \in \mathbb{K},$$
 and define the inner product by

\noindent 
\[\left\langle A,B\right\rangle =\sum^N_{i,j=1}{{\gamma }_{ij}{\alpha }_i{\overline{\beta }}_j},\]

\noindent where  $\mathrm{\Gamma }=\left[{\gamma }_{ij}\right]$ is a positive definite $N\times N$-matrix. For example, if we choose the identity matrix $\Gamma=I\ {\mathrm{and}\ \ G}_{\left(i-1\right)n+j}={E_{ij}}/{\sqrt{{\varrho }_{ij}}}$ for a matrix $\mathrm{P}=\left[{\varrho }_{ij}\right]=\left[{\varrho }_i{\varrho }_j\right]$ of $n^2$ positive weights,  then we get weighted Frobenius norm ${\parallel A\parallel }^2_{F,\mathrm{P}}={\left\langle A,A\right\rangle }_{F,\mathrm{P}}$ induced by  the weighted Frobenius inner product

\noindent 
\[{\left\langle A,B\right\rangle }_{F,\mathrm{P}}=\sum^n_{i,j=1}{{\varrho }_{ij}a_{ij}{\overline{b}}_{ij}}.\]

By Lemma \ref{l24} each matrix $\left[b_{ij}\right]\ \in {\mathcal{A}}_n$ satisfies $b_{ij}={\sigma }_i-{\sigma }_j,$  where the additive constant ${\sigma }_1$ is fixed. Hence the squared weighted distance ${dist}_{F,\mathrm{P}}\left(A,{\mathcal{A}}_n\right)$ of an anti-symmetric real matrix $A$ to the space ${\mathcal{A}}_n$ of all additively consistent real matrices is equal to the  minimal value of  the quadratic function

\[f_A\left(\sigma \right)=\sum^n_{i,j=1}{{\varrho }_{ij}{(a_{ij}-{\sigma }_i+{\sigma }_j)}^2},\]

\noindent of variable  $\sigma =\left({\sigma }_1,{\sigma }_{2,}\dots ,{\sigma }_n\right)\in {\mathbb{R}}^n$ with first coordinate ${\sigma }_1$ fixed. This minimal value is attained at the unique solution ${\sigma }_2,{\sigma }_{3,}\dots ,{\sigma }_n$ of the following system of normal equations 

\begin{equation}\label{e1}
\sum^n_{j=1}{{\varrho }_j\left(a_{ij}-{\sigma }_i+{\sigma }_j\right)=0,\ \ }\ i=2,3,\dots ,n,                    
\end{equation}

\noindent with left-hand sides equal to $-\frac{1}{4{\varrho }_i}\frac{\partial f_A\left(\sigma \right)}{\partial {\sigma }_i}.$ 

\noindent 

From now on we consider only real-valued $n\times n$ --matrices and,  unless otherwise stated, always choose  the first coordinate ${\sigma }_1$ of priority vector $\sigma $ equal to $0.$ In view of the following theorem it follows that another reasonable choice for  additive constant ${\sigma }_1$ in (\ref{e1}) would be the weighted arithmetic mean of the first row of matrix $A:$
\[{\sigma }_1=\frac{\sum^n_{j=1}{{\varrho }_ja_{1j}}}{\sum^n_{j=1}{{\varrho }_j}}.\]

\begin{thm}\label{th54}
Let $A=\left[a_{ij}\right]$ be an anti-symmetric real matrix. If $\mathrm{P}=[{\varrho }_i{\varrho }_j]$ is a matrix of positive weights, then the additively consistent orthogonal approximation $A_{approx}=[{\sigma }_i-{\sigma }_j]$ of $A$ onto ${\mathcal{A}}_n$ with respect to weighted Frobenius norm ${\parallel \cdot \parallel }_{F,\mathrm{P}}$ is determined by: 

\noindent 
\begin{equation}\label{e2}
{\sigma }_i=\frac{\sum^n_{j=1}{{\varrho }_ja_{ij}}}{\sum^n_{j=1}{{\varrho }_j}}\ ,\ i=1,2,\dots ,n.                         
\end{equation}
\end{thm}

\begin{proof}  Since the orthogonal projection is determined uniquely, it is sufficient  to check that normal equations (\ref{e1}) are satisfied by given in (\ref{e2})  values  of ${\sigma }_i.$ For this purpose denote $\left|\varrho \right|={\varrho }_1+\dots +{\varrho }_n$ and note that  
\[\mathrm{\ }\sum^n_{j=1}{{\varrho }_j\left(a_{ij}-{\sigma }_i\right)=\ \sum^n_{j=1}{{\varrho }_ja_{ij}-{\sigma }_i\left|\varrho \right|=\ \ 0}\ }\] 
for these values of ${\sigma }_i.$ Moreover, by the anti-symmetry of $A$ we have $a_{jk}=-a_{kj}$, and so

\noindent 
\[\sum^n_{j=2}{{\varrho }_j{\sigma }_j=}-\sum^n_{j=2}{\frac{{\varrho }_j}{\left|\varrho \right|}\sum^n_{k=1}{{\varrho }_ka_{kj}}}=-\sum^n_{k=1}{\frac{{\varrho }_k}{\left|\varrho \right|}\left(-{\varrho }_1a_{k1}+\sum^n_{j=1}{{\varrho }_ja_{kj}}\right)}\] 

\[={\varrho }_1{\sigma }_1-\sum^n_{k=1}{{\varrho }_k{\sigma }_k=}-\sum^n_{k=2}{{\varrho }_k{\sigma }_k.}\]

\noindent Thus the last sum is also equal to 0, which completes the proof.           \end{proof}  

\noindent 

\begin{thm}\label{th55}
 Let $M=\left[m_{ij}\right]$ be a $PC$ matrix. If $\mathrm{P}=[{\varrho }_i{\varrho }_j]$ is a matrix of positive weights$,$  then the consistent approximation

\noindent 
\[M_{approx}:={\mathrm{exp} {({\mathrm{log}\ M)\ }}_{approx}\ }=\left[{{\omega }_i}/{{\omega }_j}\right]\epsilon {\mathcal{M}}_n\]

\noindent of $M$ with respect to weighted Frobenius norm ${\parallel \cdot \parallel }_{F,\mathrm{P}}$  is determined uniquely by: 

\begin{equation}\label{e3}                      
{\omega }_i={\left[\prod^n_{j=1}{{\left(m_{ij}\right)}^{{\varrho }_j}}\right]}^{{1}/{\sum^n_{j=1}{{\varrho }_j}}}\ ,\ \ i=1,2,\dots ,n.
\end{equation}
\end{thm}

\begin{proof}  Apply Theorem \ref{th54} to the anti-symmetric matrix $A=\left[a_{ij}\right]$ with $a_{ij}={\mathrm{log} m_{ij}\ }$ in order to show that the elements of  consistent orthogonal projection ${\left({\mathrm{log} M\ }\right)}_{approx}=\left[{\sigma }_i-{\sigma }_j\right]$ of  ${\mathrm{log} M\ }$ onto  ${\mathcal{A}}_n$ are determined by: 

\noindent 
\[{\sigma }_i=\frac{\sum^n_{j=1}{{\varrho }_j{\mathrm{log} m_{ij}\ }}}{\sum^n_{j=1}{{\varrho }_j}}={\mathrm{log} {\left[\prod^n_{j=1}{{\left(m_{ij}\right)}^{{\varrho }_j}}\right]}^{{1}/{\sum^n_{j=1}{{\varrho }_j}}}\ }\ ,\ \ i=1, 2,\dots ,n.\]

\noindent Hence we get formulae (\ref{e3}) from identity ${\omega }_i={\mathrm{exp}\ {\sigma }_i,\ }$ which is a direct consequence of Definition \ref{aprox}.                                                             
\end{proof}

The direct corollaries of Theorems \ref{th54} and \ref{th55} are the following generalizations of Theorem \ref{th52}, which state that Definition \ref{aprox} is idempotent:

\begin{cor} Let $A=\left[a_{ij}\right]$ be an anti-symmetric matrix. If $\mathrm{P}=[{\varrho }_i{\varrho }_j]$ is a matrix of positive weights$,$  then the additively consistent approximation with respect to weighted Frobenius norm 
${\parallel \cdot \parallel }_{F,\mathrm{P}}$  is idempotent:

\noindent 
\[{\left(A_{approx}\right)}_{approx\ }=A_{approx}.\] 
\end{cor}

\begin{cor} Let $M=\left[m_{ij}\right]$ be a $PC$ matrix. If $\mathrm{P}=[{\varrho }_i{\varrho }_j]$ is a matrix of positive weights$,$  then the consistent approximation with respect to weighted Frobenius norm ${\parallel \cdot \parallel }_{F,\mathrm{P}}$  is idempotent:

\noindent 
\[{\left({\mathrm{exp} \left[{({\mathrm{log}\ M)\ }}_{approx}\right]\ }\right)}_{approx.}={\mathrm{exp} \left[{({\mathrm{log}\ M)\ }}_{approx}\right]\ }.\] 
\textbf{}
\end{cor}

\noindent This means that in a weighted Frobenius norm the consistent approximation mapping $M_{proj}:M\mapsto M_{approx}$ from Definition \ref{aprox} is a projection of the set ${PC}_n$ of all $PC$ matrices onto the multiplicative group   ${\mathcal{M}}_n=\left({\mathcal{M}}_n,\cdot \right).$

\subsection{Nonlinear consistent projection in weighted Frobenius norms}

\noindent 

 The squared weighted Frobenius distance ${dist}_{F,\mathrm{P}}\left(M,{\mathcal{M}}_n\right)$ of a $PC$  matrix $M$  to the space ${\mathcal{M}}_n$ of all multiplicatively consistent matrices $\left[{x_i}/{x_j}\right]$  is determined by a point $x=\left(x_1,x_{2,}\dots ,x_n\right)$ with the first coordinate $x_1=1,$  for which  minimal value of  the function

\[g_M\left(x\right)=\sum^n_{i,j=1}{{\varrho }_{ij}{(m_{ij}-\frac{x_i}{x_j})}^2},\ \ x_1=1.\]

\noindent is attained. If $\mathrm{P}=[{\varrho }_{ij}]$ is a symmetric matrix of positive weights$,$  then this minimal value is attained at solution $x_2,x_{3,}\dots ,x_n$ of the following system of nonlinear normal equations 

\noindent 

\begin{equation}        
\frac{1}{x_i}\sum^n_{j=1}{{\varrho }_j\left[\frac{x_j}{x_i}\left(\frac{1}{m_{ij}}-\frac{x_j}{x_i}\right)-\frac{x_i}{x_j}\left(m_{ij}-\frac{x_i}{x_j}\right)\right]=0,\ \ }\ \ i=2,3,\dots ,n,
\end{equation}  

\noindent 

\noindent where the left- hand sides are equal to $-\frac{1}{4}\frac{\partial g_M\left(x\right)}{\partial x_i}.$

\noindent 

 It seems unlikely that one can find an explicit solution of this system. However, it can be solved  by the locally convergent Newton's method. As a starting point, the priority vector  $x=\left(x_1,x_{2,}\dots ,x_n\right)$, given in Theorem \ref{th55}, should be used. Moreover, further improvement could be made by applying recent results on classical discrete orthogonal
 polynomials proposed in \cite{RS}. 
 
 The lack of an explicit solution should not be a huge surprise. Similar situation exists in physics with the three body problem having only numerical solution and a proof that the general case of this problem has no analytical solution. Evidently, the numerical solution is sufficient to conquer the space.

\section{Conclusions}
\label{concl}
The primary goal of this study was to generalize orthogonal projections for computing approximations of inconsistent PC matrices from the Euclidean space to the Hilbert space  of PC matrices endowed in a different inner product. However, a side product of our study seems to be even more important: there is no mathematical reasoning to support any belief that there is only one approximation method of inconsistent PC matrices. It is a matter of an arbitrary choice of the dot product for the orthogonalization projection process. However, there is a practical reason to use a Frobenius inner product (which generates GM solution), which is its computational simplicity.

\section*{Acknowledgments}
The authors would also like to express appreciation to Tiffany Armstrong (Laurentian University, Computer Science), and Grant O. Duncan, Team Lead, Business Intelligence, Integration and Development, Health Sciences North, Sudbury, Ontario, Canada) for the editorial improvements of our text and their creative comments. The research of the third author was supported by the National Science Centre, Poland as a part of the project no. 2017/25/B/HS4/01617, and by the Faculty of Applied Mathematics of AGH UST within the statutory tasks subsidized by the Polish Ministry of Science and Higher Education, grant no. 16.16.420.054.

\end{document}